\documentstyle[extras,11pt]{article}

\newcommand{\qq}{\mbox{\boldmath $q$}}
\newcommand{\rr}{\mbox{\boldmath $r$}}
\newcommand{\xx}{\mbox{\boldmath $x$}}
\newcommand{\zz}{\mbox{\boldmath $z$}}

\newcommand{\pp}{\mbox{\boldmath $p$}}

\newcommand{\ww}{\mbox{\boldmath $w$}}

\newcommand{\yy}{\mbox{\boldmath $y$}}

\newcommand{\ga}{\mbox{$\gamma$}}
\newcommand{\al}{\mbox{$\alpha$}}

\newcommand{\argmax}{\mbox{\rm argmax}}

\newcommand{\ra}{\rightarrow}

\newcommand{\CM}{\mbox{${\cal M}$}}

\newcommand{\R}{\mbox{\rm\bf R}}
\newcommand{\Rplus}{\R_+}

\makeatletter
\def\fnum@figure{{\bf Figure \thefigure}}
\def\fnum@table{{\bf Table \thetable}}

\long\def\@mycaption#1[#2]#3{\addcontentsline{\csname
 ext@#1\endcsname}{#1}{\protect\numberline{\csname
  the#1\endcsname}{\ignorespaces #2}}\par
     \begingroup
       \@parboxrestore
          \small
       \@makecaption{\csname fnum@#1\endcsname}{\ignorespaces
#3\endgroup}
      }

\begin{document}

\title{Equilibrium Pricing of Semantically Substitutable Digital Goods}

\author{
{\Large\em Kamal Jain}\thanks{Ebay Research, 2065 Hamilton Ave, San Jose,
CA 95125. Email: kamaljain@gmail.com.} 
\and
{\Large\em Vijay V. Vazirani}\thanks{College of Computing,
Georgia Institute of Technology, Atlanta, GA 30332--0280.
Supported by NSF Grants CCF-0728640 and CCF-0914732, ONR Grant N000140910755, and a Google Research Grant.
Email: {\sf vazirani@cc.gatech.edu}}
}

\setcounter{page}{0}

\maketitle

\begin{abstract}
The problem of arriving at a principled method of pricing goods and services was very satisfactorily solved for conventional
goods; however, this solution is not applicable to digital goods. This paper studies pricing of
a special class of digital goods, which we call {\em semantically substitutable digital goods}. 
After taking into consideration idiosyncrasies of goods in this class, we define a market model for it, together with
a notion of equilibrium. We prove existence of equilibrium prices for our market model using Kakutani's fixed point theorem.

The far reaching significance of a competitive equilibrium is made explicit in the Fundamental Theorems of Welfare Economics.
There are basic reasons due to which these theorems are not applicable to digital goods. This naturally leads to the question of 
whether the allocations of conventional goods are rendered inefficient or ``socially unfair'' in the mixed economy we have proposed. 
We prove that that is not the case and that in this sense, the intended goal of Welfare Economics
is still achieved in the mixed economy.
\end{abstract}

\newpage

\section{Introduction}

The central problem of mathematical economics -- arriving at a principled method of pricing goods and services which
leads to an efficient allocation of scarce resources among alternative uses -- was very satisfactorily solved for conventional
goods. The solution was based on Adam Smith's principle of maintaining parity between supply and demand \cite{Adam}, Walras' notion of
equilibrium \cite{walras}, and the Arrow-Debreu Theorem, which proved the existence of equilibrium in a very general model of the economy \cite{AD}.

By a {\em digital good} we mean a good that can be stored on a digital computer.
For such goods, which already occupy a large fraction of the economy and are expected to only grow in the future,
this solution is not applicable -- once produced, a digital good can be reproduced at (essentially) zero cost, thus making its
supply infinite. Several economists have addressed this dilemma and have proposed ways of pricing digital goods e.g., see 
\cite{digital, Bakos, Quah, varian.price}.
However, to our knowledge, none of them has given a reasonable market model that leads to pricing such goods 
from a notion of equilibrium that balances supply and demand, i.e., as was achieved for conventional goods. 

In this paper, we propose such a market model for the special class of digital goods defined in the next section. 
We start by observing that traditional market models are not applicable for pricing digital goods. The reason is that an allocation made
at a competitive equilibrium in such a model satisfies the first welfare theorem and is therefore Pareto optimal. On the other hand, since 
the supply of digital goods is unbounded, any Pareto optimal allocation must give one copy of each digital good to each agent. 
Clearly, such an allocation is not particularly meaningful for our purposes.

Throughout this paper, we will assume that copyright laws are strictly enforced. Hence, only the agent who owns/produces a certain
digital good is allowed to sell copies of it; an agent who buys this good is not allowed to sell copies of it, i.e., the owner of this good 
has a monopoly on its sale. At first sight it would appear that pricing of digital goods should follow from monopolistic considerations.
Yet, we identify a class of digital goods below for which monopolistic pricing should not apply. We then give a method of expressing
preferences that is more appropriate than the traditional method for this class of goods -- it involves expressing preferences by combining
a cardinal component with an ordinal component. This and other considerations lead to the natural assumption of uniform pricing, which in
turn yields a simple market model for which we can define notions of ``supply'' and ``demand,'' and a notion of equilibrium prices.

Perhaps the most intriguing aspect of this work is that the property of unlimited supply turns out to be an advantage rather than an 
impediment to finding equilibrium pricing of digital goods -- the pricing structure we obtain is in fact simpler than the traditional one. 
However, it is not applicable to conventional goods,
except in some very special circumstances as mentioned below. In retrospect, the two main ideas that led
to this model are our method of expressing preferences and the assumption of uniform pricing -- one can see that these two ideas
fit with each other very well.

The generality of the Arrow-Debreu theorem is quite remarkable:  for any market model for which one can define a preference
order over all bundles of goods, one can establish existence of equilibrium as a corollary of this theorem. In Section \ref{sec.model} we 
show that such a preference order cannot be defined for our market model and hence we need to prove the existence of equilibrium 
for it from first principles.

A key question remaining is whether the Fundamental Theorems of Welfare Economics \cite{Arrow, Debreu} are satisfied by this mixed economy, 
consisting of conventional 
and digital goods. Traditional market models satisfy the First Welfare Theorem, which states that allocations made via a competitive equilibrium are 
efficient.  As stated above, the availability of infinite number of copies of each digital good renders the usual notion of efficiency inapplicable to our model. 
Inefficiency would be an even more acute issue if the allocation of conventional goods, which are presumably more important for survival,
is also rendered inefficient by the mixed economy.  Hence, 
one may ask if a competitive equilibrium in this mixed economy ensures efficiency as far as conventional goods are concerned.

Our answer to this question is ``Yes''. In Section \ref{sec.welfare}, we give the notion of {\em partial Pareto optimality}, which is essentially Pareto optimality 
assuming that the digital part of the economy, i.e., production and allocation of digital goods, is fixed. As our First Welfare Theorem,
we prove that equilibrium allocations in our mixed economy are partially Pareto optimal.

The traditional Second Welfare Theorem has profound implications. It says that the issues of economic efficiency and equity 
can be effectively separated out; the distributive objectives of the state can be met via initial lump sum taxes and incentives followed by
letting the market attain its own equilibrium. Once again, in our mixed economy, this theorem is not applicable as such.
Fortunately, such a theorem is not essential for ensuring social equity -- after all, our digital economy consists of goods primarily
from the entertainment industry which are (presumably) not as essential for people's survival. 

In our Second Welfare Theorem, the target distributive objectives are provided by a restriction of the mixed economy to only conventional 
goods. We show how any target set in this manner leads to a prescription of lump sum transfers. Once they are made and the mixed economy is
allowed to attain its equilibrium, we guarantee that each agent will be at least as
happy with their allocation as was desired; of course, optimal bundles now consist of conventional and digital goods. 
The choice of `what is optimal' is left to individual agents -- if someone requires lots of inspiration and only a little bit of food 
for survival then that is what they will get. We note that whereas the Second Welfare Theorem is usually proved using the hyperplane separation
theorem, we seem to require the full power of Kakutani's fixed point thoerem.

\section{Semantically Substitutable Digital Goods}
\label{sec.semantic}

In this paper, we will study pricing of digital goods that can be partitioned into a few categories through semantic considerations such that
there is a high degree of substitutability of goods within a category. However, if all the goods in the category were equally accessible, agents 
would want to optimize all the way and get their most favorite goods. The entertainment industry is replete with such digital goods, e.g.,
music, movies and video games (barring exceptional cases).
Examples of semantic categories for music may be contemporary classical music, premier jazz, second tier jazz, 
children's songs, etc. We will call this class of goods {\em semantically substitutable digital goods}. 
Henceforth, a unit of good in a digital category will be generically called a {\em song}. The assumption of high degree of substitutability 
ensures that owners of songs cannot exercise monopolistic pricing for the sale of their songs.

Examples of digital goods that do not fall into this class include software, online news and news papers,
financial advice and audio books. Finally observe that even ideas can be viewed as digital goods, since they can be stored on a digital
computer and sold to an unbounded number of times; however, they do not fall in our class.

A common feature of semantically substitutable digital goods -- and in fact many other types of digital goods --
which sets them apart from conventional goods, is the staggering number of different goods, belonging to 
the same genre, that are available with equal ease -- and typically great ease, i.e., simply downloading from the Internet. For example,
iTunes currently has 11 million songs, the Netflix on-demand Internet streaming video store has over 100,000 movies and 
the iPhone apps store has 350,000 apps for sale. Henceforth we will assume each category in our class contains a large number of songs.

Because of this fact, viewing all the songs as individual goods and providing a traditional utility may not be feasible for the agents. Therefore, 
we first give a method of expressing preferences that is more appropriate for this class of goods. It has two components, a cardinal 
component that helps determine how many songs the agent should get from each category and an ordinal component which helps
determine which particular songs she should get from each category. For the cardinal component each agent picks a concave function,
and hence this part satisfies decreasing marginal utility, as usual. For the ordinal component, each agent provides, for each digital category,
a total order on all songs in it. This helps avoid the cumbersome (and perhaps even meaningless) task of making fine quantitative distinctions 
between the large number of
songs in each category. In the past these two components have been used separately, the best examples being traditional utility functions and 
the stable marriage problem, respectively. However, to our knowledge they have not been used together.

We now give an illustrative example. Suppose an agent needs some songs every month from each of the following categories:  jazz, pop music, 
children's music, and recently released movies. Each one has a different purpose: providing relaxation, listening while driving from home to 
work and back, liked by the children, and  entertainment during weekends. Clearly, she can make quantitative statements about the relative 
importance of each category. However, within a category, she obviously prefers to get her most favorite songs. Fortunately, unlimited supply of 
digital goods makes the latter possible.

This method of expressing preferences leads naturally to the assumption given below on pricing goods in a category. First, we note that
pricing all goods in a category individually would be an enormous burden on the seller. More importantly, it would be an even bigger burden on the 
buyer (imagine deciding between song A for \$1.23, song B for \$1.42, ...), and would discourage people from participating in such a market\footnote{Indeed, 
this is the quintessential scenario where the usual game-theoretic assumptions of individual rationality and infinite computing power 
(together with infinite patience) fail completely.}.

{\bf  Assumption of uniform pricing:} For the class of semantically substitutable digital goods, in each category, all songs will have the same price.

A few categories with uniform pricing within each category has already emerged as the norm over the last few decades in several multi-billion 
dollar marketplaces associated with digital goods: First for goods that were the forerunners of semantically substitutable digital goods 
digital, e.g., CDs for music and video cassettes and DVDs for movies. Uniformity became even more strict when these goods started being sold 
as bona fide digital goods, e.g., iTunes started off by uniformly pricing all songs, but recently partitioned their 11 million songs into 
3 price categories. For renting movies, iTunes has only 2 different price categories. Similar practices apply to video games, and in fact to
other classes of digital goods as well, e.g., iPhone/Android apps.
In contrast, uniform pricing within a few categories has not emerged as the norm for the pricing of audio books, which are digital goods. Interestingly, this
was the case for the forerunners of these goods as well, i.e. for normal books.

We note that the simplicity of pricing made possible by the assumption of uniform pricing has led to their use in some special non-digital markets
as well, e.g., for a large number of drugs, Walmart charges \$4 for a 30-day dose and \$10 for a 90-day dose. A second category, containing much fewer
drugs go for \$9 for a 30-day dose and \$24 for a 90-day dose \cite{mart}.
As another example, consider the pricing of Chinese dim sum -- the amount charged
from a table depends on the number of small plates, large plates and baskets of food consumed. These 3 categories contain different sets of food items.

In order to achieve a simplified pricing structure, can the ideas stated above be used for conventional goods after partitioning them  
into a few categories, e.g., vegetables, fruits, breads and beverages?  Can one enforce uniform prices in the categories?  Since these questions will
reveal important aspects of our market model, let us explore them in some depth. Let us assume that the agents provide cardinal utilities across
categories and ordinal utilities within categories. For the latter, for each category, say fruits, the agent provides an infinite sequence of fruits,
with repetitions. If at given prices, the agent is to receive $k$ fruits, then she should be allocated the first $k$ fruits in her sequence. Clearly,
such a market is bound for failure, since the demand for individual fruits will typically not be compatible with supply -- certain fruits will be
deficient while others will be in excess. In the absence of the capability of lowering the price of the latter, they will simply rot.

On the other hand digital goods will never be deficient because they can be replicated indefinitely. Moreover, digital goods are not perishable 
and their storage does not require huge warehouses. Thus, in a sense the digital economy appears to be tailor made
for the simplified way of expressing preferences and the simplified pricing structure presented above.

\subsection{Some idiosyncrasies of ``songs''}
\label{sec.idio}

Having settled on uniform pricing of all songs within a category, we are still left with the questions of defining the notion of an 
optimal bundle of goods for an agent and of pricing individual categories from
first principles, i.e., from a notion of equilibrium that is appropriate for our setting.
In this section, we point out some idiosyncrasies of the digital class we have identified which
will lead naturally to the formal model defined in the next section and a notion of equilibrium for it.

First, as already pointed out above, once produced, the supply of a song is unbounded.
Second, an agent desiring 2 songs wants 2 different songs; 2 copies of the same song are no better than 1 copy. Third, although all 
songs in a category are uniformly priced, they are not equally desired by an agent. In fact, the typical case is that there is a wide variation
in how an agent rates the different songs of a category. Furthermore, the ratings of different agents for songs within the same category
also vary widely. To capture the last aspect, we need to make a radical change from traditional market models as detailed in the next section.

\section{The Market Model}
\label{sec.model}

The digital economy we describe can be added to, and will harmoniously co-exist with, a conventional Arrow-Debreu economy. 
In interest of highlighting the main new ideas, we have kept the latter economy simple by assuming that it has only one good which we 
will call {\em bread}. Another simplification made is that the initial endowment of an agent will only consist of songs. 

Our model has a set $A$ of $n$ agents numbered 1 to $n$ and some genres of semantically substitutable digital goods, e.g. music 
and movies. We will assume that this set of digital goods have been partitioned into $g$ semantic categories, numbered 1 to $g$. 
We will assume that bread is numbered as good 0. Let $G = \{0, 1, 2, \ldots, g \}$ represent the set of all digital categories and bread.
We will assume that all the goods, conventional and digital, are divisible\footnote{Without the assumption of
divisibility, equilibrium will simply not exist, e.g., see \cite{gean} which assumes, in the context of collateral equilibrium,
that houses are divisible. Typical semantically substitutable digital goods are bought in large numbers, thereby making 
rounding errors insignificant and justifying the assumption of divisibility in our model.}.
We will index agents by $i$ and elements of $G$ by $j$.

Assume that the price per unit of bread is denoted by $p_0$ and the price per song in category $j$ is denoted by $p_j$. Let 
$\pp = (p_0, p_1, \ldots, p_g)$ denote the entire price vector.
We will assume that each digital category $j$ 
starts off with a non-empty set $S_j$, which in a sense ``defines'' this category. Each song in $S_j$ is owned by a unique agent; 
let $S_{ij}$ denote the songs of this category owned by $i$. Observe that this also defines the initial endowments of agents since
they consist only of songs. The earnings from the sales of a song in $S_j$ go to the specified owner.  

For each agent, the set of feasible production points is a compact, convex set in $\Rplus^{g+1}$. 
We will assume that the rate of production of bread by an agent is never zero.
Moreover, we will assume that the only factor limiting its production is labor, i.e., there is an indefinite supply of raw materials. 
We will denote the entire production by $\yy$. Set $y_{ij}$ will denote the songs in category $j$ produced by agent $i$; $y_{ij}$ may contain 
a fractional song. The entire production of $i$ will be denoted by $\yy_i$. 
The earnings of $i$ will be $\sum_{j=0}^k {p_j \cdot z_{ij}}$, where $z_{i0} = y_{i0}$ is the number of units of bread produced by $i$ and
for $1 \leq j \leq g$,  $z_{ij}$ is the number of copies of songs in $S_{ij} \cup y_{ij}$ that were sold (taking fractions into consideration).
Observe that each agent will have positive earnings since she has the fall-back option of producing bread.

At prices $\pp$, agent $i$'s optimal bundle of goods is obtained via a 2 step process. Let us introduce some definitions before giving 
this process. Function $f_i: \Rplus^{g+1} \ra \Rplus$ denotes the {\em coarse utility function} of agent $i$;
$f_i$ is assumed to be quasi-concave, continuous and satisfying non-satiation. A {\em coarse allocation} for agent $i$ is a vector
$\zz_i =(z_{i0},\ldots, z_{ig})$ specifying the number of units of bread and the number of songs from each category that get allocated to 
agent $i$. Her utility from this allocation is $f_i(\zz_i)$. 
For each digital category $j$, each agent $i$ has a total order over the set $(S_j \cup A)$; we will denoted this total order by $T_{ij}$. 
Informally, this total order specifies the rating of $i$, over all agents in $A$ as producers of the songs in category $j$, as well as the songs 
already present in $S_j$.

At prices $\pp$, agent $i$'s optimal bundle of goods is determined as follows. First, let $\zz_i =(z_{i0},\ldots, z_{ig})$ be
the coarse allocation that maximizes $f_i(\zz_i)$ w.r.t. prices $\pp$, taking into consideration the amount of money agent $i$ has available. 
Next, $i$ is allocated $z_{i0}$ units of bread and 
for each digital category $j$, $i$ will be allocated $z_{ij}$ songs from this category in the order specified by $T_{ij}$, after taking into 
consideration the songs produced by all agents in this category. Let us give a small example: Suppose agent $T_{ij}$ has song $s$
and agent $i'$ as its top 2 elements, $z_{ij} = 3.5$ and agent $i'$ has produced 4 songs in category $j$. Then, $i$ will be allocated
a copy of $s$ and copies of 2.5 of the 4 songs produced by $i'$. The eventual set of goods allocated to an agent will be called 
her {\em detailed allocation}. We will denote $i$'s detailed allocation by $\xx_i =(x_{i0},\ldots, x_{ig})$, where $x_{i0} = z_{i0}$ and
for $1 \leq j \leq g$, $x_{ij}$ is the set of songs allocated to $i$, taking fractions into consideration. Observe that an agent may buy a song 
she has produced or owns. Since the earnings from this sale go back to her, effectively she is getting such a song for free, as is reasonable.

We now give an example to show that not every bundle of songs is valid in our model and hence 
one cannot define a preference order over all bundles in our model.
Suppose $S_j = \{ s_1, s_2, s_3 \}$ and the total order of agent $i$ for this category has $s_3, s_2, s_1$ as its top 3 entries, in that
order. Then, $\{ s_1, s_2 \}$ is not a valid bundle for $i$ whereas $\{ s_2, s_3 \}$ is.

\section{The Notion of Equilibrium}
\label{sec.notion}

First, consider the following viewpoint for the case of conventional goods. The usual assumption of divisibility of conventional 
goods is not true at some minute enough level -- let us call it a morsel. Clearly, each morsel can be consumed by at most one agent. Now
equilibrium demands that all morsels that are available be consumed and no agent has unfulfilled demand for more at given prices.

Analogously, for our realm, equilibrium should satisfy the following.
Each song can be sold at most once to each agent. Furthermore, each song that is available must
be sold to at least one agent and no agent should have unfulfilled demand in any category.

We now formally define the notion of equilibrium for our model. We will make the usual assumption that each agent works in a manner 
that optimizes  her earnings and she spends all her earnings to buy an optimal bundle of goods. 
We will say that prices $\pp$, detailed allocation $\xx$ and production $\yy$ constitute an equilibrium if
they satisfy the following {\bf conditions}:
\begin{enumerate}

\item
For each agent $i \in A$, $\yy_i$ optimizes $i$'s total earnings, i.e., from sales of the songs she owns and produces and 
from bread she produces, w.r.t. prices $\pp$, the allocations of digital goods to all the agents, $\xx_{(-0)}$, and 
the production of digital goods of all the rest of the agents, $\yy_{(-i, -0)}$. 
\item
For each agent $i \in A$, $\xx_i$ is an optimal detailed allocation at prices $\pp$ for the amount of money she has earned.

\item
The market clears. For bread this means that the amount of bread demanded equals the amount of bread produced. For each digital 
category $j$, if $p_j >0$, then at least one full 
copy of each song in $S_j$ and each song produced in this category is sold. By ``full copy'' we mean that some agent buys the entire song; 
2 agents buying half of this song is not sufficient. In particular, if an agent produces a fraction $f$ of a song, then some agent must buy 
this entire fraction. Furthermore, for market clearing we want that for each agent $i$, $x_{ij}$ must not be greater than the total number of 
songs available in category $j$. 
\end{enumerate}

Observe that for the case of produced songs, optimality of agents' production also implies that one full copy of the song be 
sold -- because otherwise, the agent is better off producing bread instead of the unsold part of the song. However, this condition 
is also a market clearing condition, hence we have included it there as well.

For digital categories, our equilibrium notion has a Cournot-type aspect to it.
Unlike bread, for which the earnings depend only on the amount produced (assuming prices are fixed), the earnings from a  
song depend on the number of copies sold which in turn depends on the production of the rest of the agents.

\section{Existence of Equilibrium}
\label{sec.proof}

Let $(\pp, \xx, \yy)$ be prices, detailed allocations and production of goods for the instance $\CM$ of our market.
We first give some definitions w.r.t. digital category $j$.
With a slight abuse of notation, for an agent $i$ and category $j$, we will denote the number of songs produced by $i$ in this category by $y_{ij}$ and
for $k \in S_j$, we will define $y_{kj} = 1$, i.e., $S_j$ contains complete songs. 
Let $Y_j$ denote the total number of songs available in this category, i.e., either initially or produced. Clearly, $Y_j = \sum_{k \in (A \cup S_j)} {y_{kj}}$.
Again, with a slight abuse of notation, we will denote the number of songs obtained by $i$ from category $j$ by $x_{ij}$.
$x_{ij}$ is further partitioned by adding a third coordinate $k$. For $k \in S_j$, $x_{ijk}$ will denote how much of this song is bought by $i$. Also, for each agent
$k \in A$, $x_{ijk}$ will denote how many songs produced by agent $k$ are bought by $i$. Observe that the two sets of $k$'s are disjoint.
Furthermore, all $x_{ijk}$s, for $k \in (S_j \cup A)$, specify $x_{ij}$. Hence, we will view $\xx$ as having 3 indices, $i, j$ and $k$.

Define the {\em excess demand} of agent $i$ for category $j$,
\[  d_{ij} = \max \{0, x_{ij} - Y_j \};\]
observe that this is the number of songs that cannot be allocated to $i$ under the given allocation and production, $\xx$ and $\yy$.
Clearly, $d_{ij} \geq 0$.

For song $k \in (A \cup S_j)$, define the {\em excess supply} of this song to be
\[ l_{kj} =  y_{kj}  -  \max_{i \in A} \{x_{ijk} \} .\]
Because of the manner in which a detailed allocation is obtained, $l_{kj} \geq 0$.

We next introduce the notion of a {\em market maker}. We will need this notion only for the purpose of proving existence of equilibrium; 
this fictitious entity is not needed in our model or our notion of equilibrium. 
If for some $i$, $d_{ij} > 0$, then the market maker will produce $\max_{i \in A} \{d_{ij} \}$
extra songs in category $j$ and will sell them appropriately to ensure that each agent $i$ gets $x_{ij}$ songs. 
Furthermore, if $l_{kj} >0$, then the market maker will buy $l_{kj}$ amount of the song produced by agent $k$ in category $j$. 

Now, let us define the total number of songs bought by all agents in category $j$,
\[ b(j) =  \sum_{i \in A} {x_{ij}} \ = \sum_{i \in A}  { \left(d_{ij} + \sum_{k \in (A \cup S_j)} {x_{ijk}} \right) }   .\]
Finally, let us define the total number of songs sold from category $j$ to be
\[ s(j)  =  \sum_{k \in (A \cup S_j)}  { \left( l_{kj}  +  \sum_{i \in A}  {x_{ijk}} \right) }  .\]

For the case of bread, denote the number of units bought by
\[ b(0)  =  \sum_{i \in A}  {x_{i0}}  \]
and the number of units produced by
\[  s(0)  =  \sum_{i \in A}  {y_{i0}}  .\]

\begin{lemma}
\label{lem.equal}
For a digital category $j$, if $b(j) = s(j)$ then there is no excess demand or excess supply w.r.t. this category, i.e.,
$\forall i \in A, \ d_{ij} = 0$ and \ $\forall k \in (A \cup S_j),  \ l_{kj} = 0$.
\end{lemma}

\begin{proof}
From the definitions of excess demand and excess supply, we have
\[  \exists i \in A \ s.t. \ \ d_{ij} > 0  \ \ \ \ \Rightarrow \ \ \ \  \forall k \in (A \cup S_j): \  l_{kj} = 0 ,\]
since buyer $i$ with $d_{ij} > 0$ must buy all $Y_j$ sings. Furthermore, 
\[ \exists k \in (A \cup S_j): \  l_{kj} > 0 \ \ \ \ \Rightarrow \ \ \ \  \forall i \in A: \ d_{ij} = 0 ,\]
since not all songs in category $j$ are fully sold. Thus, w.r.t. category $j$,
we cannot have an agent having positive excess demand and a song having positive excess supply simultaneously.

Next, equating the expressions for $b(j)$ and $s(j)$ given above and canceling common terms we get
\[ \sum_{i \in A}  {d_{ij}}  \  = \    \sum_{k \in (A \cup S_j)}  {l_{kj}}  . \]
By the previous fact and the fact that both sides of this equation are non-negative, they must both be zero. The lemma
follows from the non-negativity of $d_{ij}$'s and $l_{kj}$'s.
\end{proof}

\subsection{The correspondence and its properties}
\label{sec.corres}

In this section we will define the correspondence, $F$, on which we will apply Kakutani's fixed-point theorem to show 
existence of equilibrium. We first specify the domain over which $F$ is defined.
Let $M$ denote an upper bound on the number of units of bread or the number of songs that would be available if all agents simply
produced only one good; the number of initial songs in the $S_j$'s is also taken into consideration in obtaining $M$. Next, we will modify the 
assumption of non-satiation as follows: assume that each agent gets fully satiated for good $j$ after obtaining $1.1 M$ units of it, i.e., 
her utility function ``flattens out'' in dimension $j$ after $1.1 M$ units. This does not change the equilibrium; however, it ensures the
useful property that optimal bundles are bounded, thereby enabling us to define $F$ over a convex, compact set. Let $D$ be the
interval $[0, 1.1M]$ and $\Delta_g$ be the unit $g$-dimensional simplex. Let $s = max_{j = 1}^g \{|S_j|\}$.

Then, the domain of $F$ is $\Delta_g \times D^{n(g+1)(n+s)} \times D^{n(g+1)}$.
We will represent an element of this domain by $(\pp, \xx, \yy)$, where $\pp \in \Delta_g$, $\xx \in D^{n(g+1)(n+s)}$ and $\yy \in D^{n(g+1)}$.

Finally,  for an element $(\pp, \xx, \yy)$ in this domain, $(\pp', \xx', \yy') \in F(\pp, \xx, \yy)$ iff $(\pp', \xx', \yy')$ lies in this domain and

\begin{itemize}
\item
$\xx'$ is an optimal detailed allocation w.r.t. prices $\pp$ and production $\yy$.
\item
$\yy'$ is a production of all agents satisfying: $\forall i \in A$, $\yy_i '$ is an optimal production of agent $i$ w.r.t. prices $\pp$,
allocation $\xx$ and the productions of the rest of the agents according to $\yy_{-i}$.
\item
\[ \pp'  =  \argmax_{\qq \in \Delta_g} \ \left\{ q_0 {{b(0) + 1} \over {s(0) + 1}}  +   {\sum_{j=1}^g  {q_j {{b(j)} \over {s(j)}} } } \right\}  .\]

Observe that the denominators of each of the $g+1$ terms in this sum is positive; for the first term because of adding 1 and for
the remaining terms because of the songs in $S_j$'s.
\end{itemize}

We will next prove a sequence of lemmas which will lead to Theorem \ref{thm.hemi}.
For this purpose, it will be convenient to decompose $F$ into three correspondences, $F_1, F_2, F_3$ which will
map $(\pp, \xx, \yy)$ to the set of $\pp'$s, $\xx'$s and $\yy'$s respectively, as defined above. Also, define $F_2'$ to be
the correspondence that maps $(\pp, \xx, \yy)$ to the set of $\zz'$s, where $\zz'$ is an optimal coarse allocation at prices
$\pp$, given production $\yy$.

\begin{lemma}
\label{lem.one}
Correspondence $F_2'$ is upper hemi-continuous.
\end{lemma}

\begin{proof}
First observe that for each buyer $i$, the earning of $i$ is a continuous function of $(\pp, \xx, \yy)$.
Since $i$'s coarse utility function, $f_i$, is concave and continuous, the set of $\zz_i'$s is an upper hemi-continuous
correspondence of $\pp, \ \yy$, and the earnings of $i$. Composing the two we get that the
correspondence $F_2'$ is upper hemi-continuous.
\end{proof}

\begin{lemma}
\label{lem.two}
Correspondence $F_2$ is upper hemi-continuous.
\end{lemma}

\begin{proof}
For conventinal goods, the detailed allocation is the same as the coarse allocation. For a digital category, the coarse allocation is a
number; this, together with the total order, yields a vector of dimension $n + s$, i.e., the detailed allocation. Clearly, this map is
continuous. Hence, the detailed allocation is a continuous function of the coarse allocation. This fact together with 
Lemma \ref{lem.one} yields the current lemma.
\end{proof}

\begin{lemma}
\label{lem.three}
Correspondence $F_3$ is upper hemi-continuous.
\end{lemma}

\begin{proof}
This follows from the assumption that for each agent, the set of feasible production points is a compact, convex set. 
\end{proof}

\begin{lemma}
\label{lem.four}
The set of songs produced by the market maker and the set of songs bought by the market maker are continuous 
functions of the coarse demand and the production of all agents.
\end{lemma}

\begin{proof}
The lemma follows from the following two observations:
\begin{itemize}
\item
For each agent $i$ and category $j$, the excess demand of $i$ for $j$ is a continuous 
function of the coarse demand and the production of all agents.
\item
For each song $k$ and category $j$, the excess supply of $k$ is a continuous 
function of the coarse demand and the production of all agents.
\end{itemize}
\end{proof}

Given functions $b, s: \ \{0, 1, \ldots g \} \ \ra \Rplus$ define correspondence $F_1'$ to be the set of all $\pp'$ satisfying:
\[ \pp'  =  \argmax_{\qq \in \Delta_g} \ \left\{ q_0 {{b(0) + 1} \over {s(0) + 1}}  +   {\sum_{j=1}^g  {q_j {{b(j)} \over {s(j)}} } } \right\}  .\]
Note that $b$ and $s$ are meant to be the functions defined above, which represent the total number of songs 
(amount of bread) bought and sold, respectively.

\begin{lemma}
\label{lem.five}
Correspondence $F_1'$ is upper hemi-continuous.
\end{lemma}

\begin{proof}
Observe that $\pp'$s are the optimal solutions to an LP whose objective function is obtained from $b$ and $s$
but whose constraints are fixed, i.e., $\pp' \in \Delta_g$. Upper hemi-continuity follows from the fact that the feasible
region of this LP is a compact set.
\end{proof}

\begin{lemma}
\label{lem.six}
Correspondence $F_1$ is upper hemi-continuous.
\end{lemma}

\begin{proof}
First consider the function $b$ representing the total number of songs (amount of bread) bought, given allocations
$\xx$ and production $\yy$. It is easy to see that Lemma \ref{lem.four} implies that the set of all such $b$s is an 
upper hemi-continuous correspondence. A similar statement holds for the set of all functions $s$.

Now, Lemma \ref{lem.five} and the fact that the composition of upper hemi-continuous correspondences is an
upper hemi-continuous correspondence proves the current lemma.
\end{proof}

\begin{theorem}
\label{thm.hemi}
Correspondence $F$ is  upper hemi-continuous and has nonempty, compact and convex values.
\end{theorem}

\begin{proof}
Non-emptiness is obvious and
upper hemi-continuity of $F$ follows from Lemmas \ref{lem.two}, \ref{lem.three} and \ref{lem.six}. 

We will prove that the correspondences $F_1, F_2$ and $F_3$ have convex values, hence proving it for $F$. $F_1$ is convex-valued because the set of prices are
solutions to an LP, as shown in Lemma \ref{lem.five}. Next, observe that correspondence $F_2'$ has convex values, since
coarse allocations maximize a concave function, i.e., the utility function, subject to budget constraints. It is easy to see
that the mapping from coarse allocations to detailed allocations preserves this convexity and hence $F_2$ is convex-valued.
Finally, optimal production is maximizing a linear function, i.e., the earnings of each agent, subject to being in the convex
set of feasible production points. Hence, the set of optimal productions also form a convex set and hence $F_3$ is convex-valued.

Finally, compactness of the values of $F$ follows from its upper hemi-continuity.
\end{proof}

\subsection{The fixed points of $F$ are equilibria}

We show below that the set of fixed points of $F$ captures exactly the set of equilibria of market $\CM$.

\begin{theorem}
\label{thm.exists}
$(\pp^*, \xx^*, \yy^*)  \in  F(\pp^*, \xx^*, \yy^*)$ iff $(\pp^*, \xx^*, \yy^*)$ is an equilibrium of market $\CM$.
\end{theorem}

\begin{proof}
($\Rightarrow$) \  The definition of correspondence $F$ implies the first 2 conditions of equilibrium, i.e., allocation and production are optimal.
Hence we only need to establish the market clearing condition.

Let
\[  \alpha  =  \max \left\{ {{b(0) + 1} \over {s(0) + 1}} , \  {{b(1)} \over {s(1)}}, \ {{b(2)} \over {s(2)}}, \ \ldots , \  {{b(g)} \over {s(g)}}  \right\} . \]
From the manner in which $\pp'$ is obtained in the definition of correspondence $F$, it follows that if 
$p^*_0 > 0$ then $(b(0) + 1) / (s(0) + 1) = \alpha$ and for $1 \leq j \leq g$, if $p^*_j > 0$ then $(b(j) / s(j))  = \alpha$.

If $p^*_0 > 0$, we get $p^*_0 b(0) = \alpha p^*_0 s(0) + (\alpha - 1) p^*_0$.  Observe that $p^*_0 b(0)$ is the total money
spent on bread and $p^*_0 s(0)$ is the total money earned from bread.

For $1 \leq j \leq g$, if $p^*_j > 0$ then $p^*_0 b(j)  =  \alpha p^*_0 s(j)$.
With the involvement of the market maker, the l.h.s. of this equation is the total money spent on digital category $j$
and the r.h.s. is $\alpha$ times the total money earned from this category. Adding up all these equations, the l.h.s. will be the
total amount of money spent by all agents under prices, allocation and production given by $(\pp^*, \xx^*, \yy^*)$.
The r.h.s. will be $\alpha$ times the total money earned by all agents; in addition, if $p^*_0 > 0$, the r.h.s. will have the additional
term $(\alpha - 1) p^*_0$.  

Now, if $\alpha > 1$, the total money spent will exceed the total money earned. However, by definition, an agent cannot spend more
than she earns, leading to a contradiction. If $\alpha < 1$, the total money spent will be strictly smaller than the total money earned.
If so, there is an agent who is not spending all the money she earns.
Because of the assumption of non-satiation, which applies while her allocation lies within $D^{g+1}$,
she is not getting an optimal allocation, leading to a contradiction. Hence we get that $\alpha = 1$.

Consider a good $j$ with $j \neq 0$ and $p^*_j = 0$. Now $b(j)/s(j)$ cannot exceed the maximum value, i.e. $\alpha$,
hence $b(j) \leq \alpha s(j)$. Since $\alpha = 1$, $b(j) \leq s(j)$. However, observe that
by non-satiation, each buyer will desire $1.1 M$ units of $j$. On the other hand, since at most $M$ units
of $j$ are available for sale, $b(j) > s(j)$, leading to a  contradiction. If $p^*_0 = 0$, we must have $b(0) + 1 \leq s(0) + 1$, leading
to the same contradiction. Therefore, for all goods $j$, $p^*_j >0$. Hence, for all goods $j$, $b(j) = s(j)$.

For the case of bread, this directly implies market clearing. For a digital category, we get from Lemma \ref{lem.equal} that there is no
excess demand or excess supply w.r.t. this category. Therefore, even without the intervention of the market maker, the market clears for this good.

($\Leftarrow$) \  Assume that $(\pp^*, \xx^*, \yy^*)$ is an equilibrium of market $\CM$.
Because of the market clearing condition, there is no excess demand or supply of any good and therefore
for each $j$, $d(j) = s(j)$. Hence any $\pp' \in \Delta_g$ maximizes the expression in the definition of $F$, as does $\pp^*$.
This together with optimality of allocations and production guaranteed by equilibrium give 
$(\pp^*, \xx^*, \yy^*)  \in  F(\pp^*, \xx^*, \yy^*)$.
\end{proof}

\section{Welfare Theorems for the Conventional Part of the Mixed Economy}
\label{sec.welfare}

In the Introduction, we have provided basic reasons due to which the Welfare Theorems are not applicable as such to digital goods. 
Below we salvage both Welfare Theorems to the extent possible -- we show that the allocations of conventional goods are not rendered inefficient 
or ``socially unfair'' in the mixed economy.

Consider an equilibrium for a market $\CM$ in our model and let $u_1, \ldots u_n$ be the utilities accrued by the agents,
as given by their coarse utility functions. Fix the production and allocation of digital goods as per this equilibrium
and consider all possible ways of changing the production of conventional goods (as long as each agent is operating as a feasible production
point) and allocation of available conventional goods among agents. If the utilities accrued in this manner never Pareto dominate 
$u_1, \ldots u_n$  then we will say that utilities $u_1, \ldots u_n$ are {\em partially Pareto optimal}.

\begin{theorem}
\label{thm.first}
{\bf (First Welfare Theorem)}
The utilities accrued in any equilibrium of our model are partially Pareto optimal.
\end{theorem}

\begin{proof}
Let equilibrium utilities be $u_1, \ldots u_n$ and suppose they are not partially Pareto optimal. Then there exist utilities
$v_1, \ldots v_n$ which are attained by changing the production and allocation of conventional goods such that w.l.o.g.,
$v_i \geq u_i$ for $1 \leq i \leq n-1$ and $v_n > u_n$. We will call this the second production and the second allocation.
Let $\pp$ denote equilibrium prices of goods and let $m_i$ be the money spent by agent $i$ for buying the bundle at prices $\pp$ that 
gives her utility $u_i$. 

Consider agent $i$, for $1 \leq i < n$. Since her equilibrium allocation is a utility maximizing bundle of goods, the cost, at prices $\pp$, 
of her second allocation must be $\geq m_i$ (since the utility she accrues from it is $v_i \geq u_i$).
By the same argument, the cost of the second allocation to agent $n$, which gives her utility $v_n > u_n$, must be $> m_i$.
Hence, the total cost, at prices $\pp$, of all goods allocated in the second allocation is $> \sum_i {m_i}$. Hence, the total value of 
conventional goods produced by the agents in the second production is more than that in equilibrium, which contradicts optimality of 
production of the agents w.r.t. prices $\pp$. Hence utilities $u_1, \ldots u_n$ are partially Pareto optimal.
\end{proof}

Let us say that the economy represented by market $\CM$ consisting of digital and conventional goods is a {\em mixed economy}.
By {\em restriction of $\CM$ to conventional economy} we mean that each agent will be involved in the production of conventional
goods only, via a feasible point in her production set, and the set of all conventional goods available as a result will be 
allocated to the agents in some manner.

Any way of redistributing the total earnings of all agents in a mixed economy will be called a {\em wealth transfer function},
and will be denoted by $\ww$. Thus, for each agent $i$, $w_i \geq 0$, and $\sum_{i \in A} {w_i}$ is the total earnings of all agents.
An {\em equilibrium with wealth transfer function} $\ww$ differs from the notion of equilibrium given in Section \ref{sec.notion} only
in that each agent now spends money $w_i$ to obtain an optimal detailed allocation; clearly, $\sum_{i \in A} {w_i}$ must be the total
earnings of all agents.

\begin{theorem}
\label{thm.second}
{\bf (Second Welfare Theorem)}
If $u_1, \ldots, u_n$ are utilities attained by agents, via feasible production and allocation, in a restriction of $\CM$ to a conventional economy, 
then there exists a wealth transfer function $\ww$ and a parameter $\al \geq 1$ such that an equilibrium with transfer function $\ww$ allocates the utility of 
$v_i = \al \cdot u_i$ to each agent $i$.
\end{theorem}

\begin{proof}
We will prove the existence of such an equilibrium and wealth transfer function using Kakutani's fixed point theorem.
For this purpose, we will define a correspondence $H$ that is related to correspondence $F$ defined in Section \ref{sec.corres}; terms 
that are not defined in this proof will have the same definition as in Section \ref{sec.proof}.

Using parameter $M$ defined in that section, and assuming that the sum of prices of all conventional goods and digital categories is unit,
one can place an upper bound on the total earnings of all agents, say $e$. Let $E$ be the interval $[0, e]$. Now, the domain of $H$ is 
$\Delta_g \times D^{n(g+1)(n+s)} \times D^{n(g+1)} \times E^n$.
We will represent an element of this domain by $(\pp, \xx, \yy, \ww)$, where $\pp \in \Delta_g$, $\xx \in D^{n(g+1)(n+s)}$, $\yy \in D^{n(g+1)}$,
and $\ww \in E^n$.

For an element $(\pp, \xx, \yy, \ww)$ in this domain, $(\pp', \xx', \yy', \ww') \in H(\pp, \xx, \yy, \ww)$ 
iff $(\pp', \xx', \yy', \ww')$ lies in this domain and

\begin{itemize}
\item
$\xx'$ is an optimal detailed allocation w.r.t. prices $\pp$, production $\yy$, and wealth transfer function $\ww$.
\item
$\yy'$ is a production of all agents satisfying: $\forall i \in A$, $\yy_i '$ is an optimal production of agent $i$ w.r.t. prices $\pp$,
allocation $\xx$, and the productions of the rest of the agents according to $\yy_{-i}$.
\item
\[ \pp'  =  \argmax_{\qq \in \Delta_g} \ \left\{ q_0 {{b(0) + 1} \over {s(0) + 1}}  +   {\sum_{j=1}^g  {q_j {{b(j)} \over {s(j)}} } } \right\}  .\]

Observe that the denominators of each of the $g+1$ terms in this sum is positive; for the first term because of adding 1 and for
the remaining terms because of the songs in $S_j$'s.
\item
Let the total earnings of all agents under prices $\pp$, allocation $\xx$, and production $\yy$ be $\gamma$. Let $\Delta_{\ga}$ be $\ga$
scaling of the unit $n-1$-dimensional simplex, i.e., the $n$ coordinates of each point should sum up to $\ga$. Then,
\[ \ww' =  \arg \min_{\rr \in \Delta_{\ga}}  \ \left\{  \sum_{i \in A}  {{ f_i(\xx'_i) \cdot r_i} \over {u_i}}  \right\}  ,\]

where with a slight abuse of notation, the argument of $f_i$ is a detailed allocation rather than a coarse allocation; the intended
meaning of $f_i(\xx'_i)$ is the utility derived by $i$ from the coarse allocation corresponding to $\xx_i$.
\end{itemize}

It is easy to see that correspondence $H$ satisfies the properties given in Theorem \ref{thm.hemi} that were satisfied by $F$ and 
hence by Kakutani's fixed point theorem, it has a fixed point, say $(\pp^*, \xx^*, \yy^*, \ww^*)$. By a proof similar to that
of Theorem \ref{thm.exists}, this must be an equilibrium with wealth transfer function $\ww^*$. 

For $i \in A$, let $v_i = f_i(\xx_i^*)$, i.e., the utility derived by $i$ in this equilibrium. Because of the way $\ww'$ is
defined in correspondence $H$, it is easy to see that there is a parameter $\al >0$ such that for each $i \in A, \ v_i = \al \cdot u_i$.
Let $n_i$ be the total earnings of agent $i$ in this equilibrium. Clearly, $\sum_{i \in A} {n_i} \ = \ \sum_{i \in A} {w^*_i}$.

Finally, we prove that $\al \geq 1$. 
Consider the production and allocation in the restriction of $\CM$ to a conventional economy that gave rise to  
utilities $u_1, \ldots, u_n$. Assuming that the prices of the conventional goods are as in $\pp^*$, let $m_i$ be the earnings of 
agent $i$ and let $s_i$ be cost of the allocation given to $i$, i.e., the money spent by $i$. 
Clearly, $\sum_{i \in A} {m_i} \ = \ \sum_{i \in A} {s_i}$. 

The crucial observation is that the production that was chosen by agent $i$ in the restricted economy
is available to agent $i$ in the mixed economy as well and we have assumed the same prices in both economies. Hence, 
her earnings need to be at least as large, i.e, $n_i \geq m_i$.
Combining with the two equations established above, we get $\sum_{i \in A} {w^*_i} \ \geq \ \sum_{i \in A} {s_i}$.
Hence, there is an agent, say $i$, such that $w^*_i \geq s_i$, i.e, $i$ spends at least as much in the mixed economy
as in the restricted economy. Since the prices of conventional goods are the same in both and in the mixed economy, 
$i$ is obtaining a bundle that maximizes her coarse utility function, $v_i \geq u_i$. Hence $\al \geq 1$.
\end{proof}

\section{Salient Features of our Model}
\label{sec.salient}

We have made some major simplifying assumptions in the preceeding sections in order to keep the formulations and proofs easy to follow.
Two of these assumptions are easy to relax conceptually and we do so here. First, we will assume a complete conventional economy in our model
rather than just one good, i.e, bread. In particular, this change is required when talking about a Pareto optimal allocation, 
since in the presence of only one
good, any allocation is Pareto optimal. Second, we do not have to assume that agent $i$ likes all songs produced by agent $i'$ in 
category $j$ equally. Let $S_{i' j}$ represent the set of all possible songs that agent $i'$ could potentially produce in category $j$.
Now, for each category $j$, each agent $i$ will pick a total order $T_{ij}$ over
\[ S_j \cup \left( \bigcup_{i' \in A} {S_{i' j}} \right) . \]
Thus, $T_{ij}$ gives the ratings of $i$ over all songs in category $j$ which either exist or can potentially be produced.  
It is easy to see that nothing changes conceptually in the proof of existence of equilibrium. Observe that each agent is required
to pick a total order over the set stated above -- this requirement is akin to a non-satiation assumption in the Arrow-Debreu model.

Due to fundamental differences between conventional and digital goods, differences between the Arrow-Debreu model and ours were
unavoidable. However, one can see that these are differences only in a very literal sense -- in spirit, we have tried to remain
as close to the Arrow-Debreu model as was possible. The notion of an Arrow-Debreu equilibrium is very much a best response of each
agent to the actions of all other agents -- the assumption of a large economy, in which individual actions do not change prices of
goods, makes the notion seem less stringent than a Nash equilibrium. It is because of this assumption that the only signal needed by 
a producer to determine his optimal production is the prices of goods. 

In our economy also, we assume that individual actions do not change prices of digital categories. Although price is not unique to individual 
songs, demand is, and demand crucially determines the earnings of a producer. Hence, for such goods, the signal needed by a producer is 
not only the prices of digital categories but also the potential sales of a song if he were to produce it in the various categories. 
This gives our notion of equilibrium a Cournot-type or Nash equilibrium-type flavor.

These differences naturally make our exact mathematical formulation quite different from that of Arrow-Debreu. In turn, the difference
manifests itself in the proof of existence of equilibrium -- the notion of a market maker in our proof helps simplify it considerably. 

In our definition of market clearing, we have required that one entire copy of each song (or fraction) that is produced must be sold to
some agent. We note that ``one'' cannot be replaced by any other integer. For instance, if we replace ``one'' by ``two'',
then in order to sell the second copy, we may need to drop the price of some category to a low value. However, at this price, some agent
may demand more songs than there are in this category, contradicting market clearing for a different reason.

In choosing an optimal production and an optimal consumption, we are assuming that agents are simply optimizing at both steps.
In particular, they are not allowed to resort to strategic behavior, i.e., pick a suboptimal production that eventually leads to
an even better bundle of goods.

Recall that an agent is allowed to buy a copy of her own song. In this process, what she earns from herself (in her optimal production),
she spends to buy her own song in her optimal bundle. Note however, that no inconsistencies are created in the process -- the agent
is still required to choose an optimal production followed by an optimal consumption, and hence clearing of her song is not automatic.

The issue of divisibility of goods is a stumbling block in the conventional economy as well. Whereas this assumption holds for consumer
goods, for a very large fraction of the economy, consisting of automobiles, houses, electronic goods, etc., this assumption does not hold.
On the other hand, without the assumption of divisibility, fixed point theorems will not be applicable, ruling out existence of equilibrium.
Divisibility of songs is a less contentious assumption than divisibility of cars -- one can enjoy half a song but not half a car.

Observe that our model does not postulate the existence of an entity, such as iTunes, for selling songs. It is simply postulating
uniform prices for all songs in a particular category. Under this assumption, the market as a whole will decide the customary price for each 
category, by taking into account supply and demand, and producers will price their songs accordingly -- if they don't, their songs
will presumably be ignored by customers who have little time for dealing with complex pricing structures. 

The central question that arises is: Won't the more talented producer of songs want to sell her songs at a premium? Besides the
issue of simplicity of prices, there are also economic reasons why she may not: First, the negligible cost of making copies in the digital 
realm ensures that a higher volume
of sales leads to substantially higher earnings than happens for conventional goods. Indeed, the standard assumption in the latter setting is
decreasing marginal production (which makes marginal cost the natural price for such goods).  
Second, a consideration that is much more important in the digital realm than in the conventional marketplace is capturing market share. 
The reason again is the limitless supply of a song, once it is produced, and the desire to capitalize on this.

If despite this, a song producer deviates from standard prices, she may have to justify this, since there would be a plethora of songs
that are priced ``right''. Currently, in the marketplace, such price differentiation and its justification typically follows technical 
considerations, such as standard vs high definition recording, standard vs 3D movie, old vs newly released movie, etc.

In our model, and in reality, once a buyer decides how many songs she desires, she simply chooses the best ones
per her taste. This automatically ensures that the most talented song producers are rewarded with higher sales and hence higher earnings.

Although our model does not allow pricing each digital good individually, since that would be prohibitively cumbersome for customers, it
does allow the possibility of splitting a category into a few categories by quality, e.g., premium songs and regular songs,
and pricing them differently. A song producer is free to put her song in either category. However, if she over-rates 
her song and places it in the premium category,  she runs the risk of being at the bottom of all agents' lists and thereby earning even less 
than she would if she had priced her song ``correctly''. 

Next we mention the mathematical reason for assuming a non-empty starting set of songs in each category. Without such a set, the market
clearing condition applies to produced songs only, but this is trivially satisfied as follows: set the price
of this category so high that no one wants a song from it and hence no one produces it. Our assumption ensures that such meaningless
``equilibria'' are ruled out. Once again, this aspect arises because of the peculiarities of the digital setting.

A potential pitfall with the assumption mentioned above is that if $S_j$ contains a song that no one desires and everyone puts at the 
bottom of their lists, then the only way to clear this song may be to set the price of this category to almost zero. The following reasonable
assumption provides a recourse: the songs in $S_j$ are special -- not only do they define the category, but they are also
among the better rated songs in this category, e.g., the set of songs defining classical music should not contain 
one that is simply noise.

Consider the very special instance in which for a certain category, all agents have the same total order. In this case, equilibrium entails
giving at least one agent all songs in this category and hence the equilibrium price of this category may turn out to be very low. 
Can this eventuality be held against our market model? Our answer is ``no'' -- even in the traditional model,
if the instance is ill-formed, the price of some goods may have to be made very low, or even zero, to clear it.
For example, assume that agents have Leontief utilities for bread and butter and they all desire equal amounts of both these items.
Now if the market has equal amounts of both items, both can be priced positively. However, even if the difference in amounts is
minuscule, the more abundant good will be priced at zero.

\section{Discussion}
\label{sec.discussion}

As pointed out in Section \ref{sec.salient}, we have made a very deliberate attempt at keeping our model simple in order to highlight its 
main new features. Exploring more complex market models with digital goods seems worthwhile. In particular, one could add complementarity
in production, e.g., the production of a piece of classical music or a movie requires people of many different talents to come together,
and one could add complementarity in consumption, e.g., if someone buys a certain movie, they would also like to buys
songs sung in that movie.

Clearly, a major open question remaining is understanding the pricing of digital goods that do not lie in our class.
An exceedingly important consideration in digital marketplaces today is ``capturing market share'' or ``capturing attention''.
Providing a better understanding of this aspect via formal models would be very worthwhile, e.g., shedding more light on
the economics of providing some digital good for free, such as web search results, in return for attention, 
which in turn results in sales of other goods, including conventional ones.

An important question remaining is the design of efficient algorithms
for computing equilibria for our model and experimentally verifying its merits and its predictive value, e.g., predicting the price 
of a newly introduced category.
For conventional goods, Fisher-based models have been found to have better algorithmic properties than Arrow-Debreu-based models,
and it may be easier to seek algorithms for the Fisher-based version of our model as well. The digital part of such a model would have 
buyers with money who desire goods, a labor force that produces these goods and categories of digital goods.

Looking beyond these immediate questions, one may attempt to include in the model mechanisms for creation of new categories and
extinction of old categories. In a different direction, it is worthwhile exploring other models that move away from game-theoretic
assumptions of full rationality and infinite computing power.

\section{Acknowledgments}
We are indebted to Professor Kenneth Arrow for providing us very insightful and detailed comments on our model, several of which
have found their way into this writeup. We also wish to thank Professors John Ledyardb and Eyal Winter, and Dr. Guilherme P. de Freitas
for valuable discussions.

\bibliography{kelly}
\bibliographystyle{alpha}

\end{document}